\documentclass[10pt]{llncs}
\usepackage[title]{appendix}
\usepackage{cmap}
\usepackage[T1]{fontenc}
\usepackage[utf8]{inputenc} 
\usepackage{adjustbox}
\usepackage{graphicx}
\usepackage{array}
\usepackage{wrapfig}
\usepackage[style=base]{caption}     
\usepackage[ngerman,english]{babel}
\usepackage[math]{blindtext}
\usepackage{cite}
\usepackage{paralist}
\usepackage{csquotes}
\usepackage{chngcntr}
\usepackage{pdfcomment}
\usepackage{ltxcmds}
\usepackage{xspace}
\usepackage{xspace}
\usepackage{pgf,tikz}
\usetikzlibrary{arrows}
\usepackage{float}
\usepackage{tikz-3dplot} 
\tikzstyle{arrow}=[draw, -latex]
\usepackage{xifthen}
\usepackage{mathtools}
\usepackage{graphicx}
\usetikzlibrary{shapes,arrows}
\usepackage[capitalise,nameinlink]{cleveref}

\begin{document}

\title{A Two Query Adaptive Bitprobe Scheme Storing Five Elements}
\author{Mirza Galib Anwarul Husain Baig, Deepanjan Kesh, \and Chirag Sodani }
\institute{Indian Institute of Technology Guwahati, Guwahati, India \\ \email \{mirza.baig,deepkesh,chirag.sodani\}@iitg.ac.in}
\maketitle

\begin{abstract}
   We are studying the adaptive \emph{bitprobe} model to store an arbitrary subset $\mathcal{S}$ of size at most five from a universe $\mathcal{U}$ of size $m$, and answer the membership queries of the form ``Is x in S?'' in two \emph{bitprobes}. In this paper, we present a data structure for the aforementioned problem. Our data structure takes $\mathcal{O}(m^{10/11})$  space. This result improves the non-explicit result by Garg and Radhakrishnan \cite{garg2014set} which takes $\mathcal{O}(m^{20/21})$ space, and the explicit result by Garg \cite{gthesis} which takes $\mathcal{O}(m^{18/19})$ space for the aforementioned set and query sizes.
  \end{abstract}


\section{Introduction}\label{sec:intro} 
In the static membership problem, we are given a universe $\mathcal{U}$ of size $m$, and our task is to design a data structure that can store an arbitrary subset $\mathcal{S}$ of size at most $n$ such that the membership queries of the form ``Is $x$ in $S$?'' can be answered correctly. We study this problem in the \emph{bitprobe} model of computation. The complexity of the static membership problem in this model is measured in terms of the size of the data structure denoted by $s$, and the number of bits of the data structure accessed denoted by $t$ . It is the later of the two properties which lend its name \emph{bitprobe} model. In this model all other operations are free. Solutions to the above mentioned problems in this model are termed as schemes. Each scheme consists of two parts, one is the storage scheme, and the other is query scheme. Storage scheme maps an arbitrary subset of cardinality  $n$ from a universe of size $m$ given to be stored to the $s$ bits of the data structure. Query scheme maps every element belonging to the universe $m$ to the $t$ locations of the data structure and it decides the membership of the query element by reading those $t$ locations. The storage and query scheme together gives a $(n,m,s,t)$-scheme which stores a set of size at most $n$ from a universe of size $m$ and uses $s$ bits in such a way that membership query can be answered in $t$ probes. This is a well studied problem over several decades and it has been discussed in \cite{buhrman2002bitvectors},\cite{radhakrishnan2001explicit1},\cite{alon2009power},\cite{radhakrishnan2010data2},\cite{lewenstein2014improved3},\cite{nicholson2013survey},\cite{garg2014set} and \cite{garg2016set}. 
\paragraph{}
A $(n,m,s,t)$-scheme is said to be adaptive if the location of the probes depends upon the bit returned by the prior probes. Whereas in a non-adaptive scheme location of the probes are fixed and it does not depend upon the bit returned by prior probes. 

\subsection{Two Adaptive \emph{Bitprobe} Model}
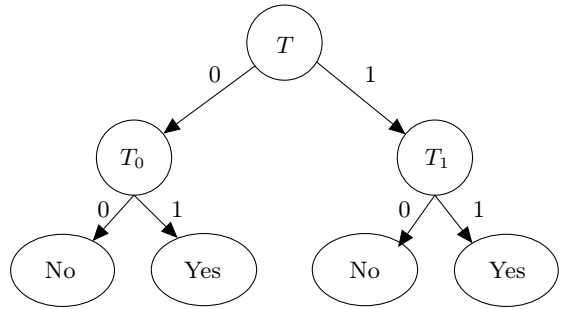
\begin{figure}[t]
	\begin{center}
		\begin{tikzpicture}[line cap=round,line join=round,>=triangle 45,x=1.0cm,y=1.0cm]
		\clip(-3.85,-3.83) rectangle (3.85,0.95);
		\draw(0,0) circle (0.5cm);
		\draw(-2,-1.53) circle (0.5cm);
		\draw(2,-1.53) circle (0.5cm);
		\draw [rotate around={0:(-2.95,-3.03)}] (-2.95,-3.03) ellipse (0.69cm and 0.48cm);
		\draw [rotate around={-180:(2.95,-3.03)}] (2.95,-3.03) ellipse (0.69cm and 0.48cm);
		\draw [rotate around={0:(-1.07,-3.02)}] (-1.07,-3.02) ellipse (0.7cm and 0.48cm);
		\draw [rotate around={-180:(1.07,-3.02)}] (1.07,-3.02) ellipse (0.7cm and 0.48cm);
		\draw [->] (-0.39,-0.31) -- (-1.61,-1.21);
		\draw [->] (0.43,-0.26) -- (1.61,-1.22);
		\draw [->] (-2,-2.03) -- (-2.55,-2.64);
		\draw [->] (-2,-2.03) -- (-1.42,-2.6);
		\draw [->] (2,-2.03) -- (1.51,-2.72);
		\draw [->] (2,-2.03) -- (2.5,-2.67);
		\draw (-0.22,0.2) node[anchor=north west] {$T$};
		\draw (-2.3,-1.3) node[anchor=north west] {$T_0$};
		\draw (1.75,-1.3) node[anchor=north west] {$T_1$};
		\draw (-3.3,-2.8) node[anchor=north west] {No};
		\draw (-1.45,-2.8) node[anchor=north west] {Yes};
		\draw (0.75,-2.8) node[anchor=north west] {No};
		\draw (2.6,-2.8) node[anchor=north west] {Yes};
		\draw (-1.12,-0.2) node[anchor=north west] {0};
		\draw (0.95,-0.2) node[anchor=north west] {1};
		\draw (-2.6,-2) node[anchor=north west] {0};
		\draw (-1.62,-2) node[anchor=north west] {1};
		\draw (1.4,-2) node[anchor=north west] {0};
		\draw (2.39,-2) node[anchor=north west] {1};
		\end{tikzpicture}
	\end{center}
	\caption{A decision tree for the two adaptive \emph{bitprobe} model}
	\label{fig1}
\end{figure}
In this section, we will discuss a two adaptive \emph{bitprobe} model in the context of two adaptive \emph{bitprobe} scheme. A two adaptive \emph{bitprobe}  scheme in this model  consist of three tables namely $T,T_0$, and $T_1$ as shown in Figure~\ref{fig1}. Furthermore, as discussed earlier the data structures in this model consist of two schemes a storage scheme and a query scheme. Storage scheme maps an arbitrary subset given to be stored to the three tables mentioned earlier.  Query scheme decides the membership of a query element by probing two location of the data structure. Given a query element, the first probe is made into the table $T$. The next query depends upon whether the bit returned by the table $T$ is zero or one. If the bit returned by the table $T$ is zero it makes next query to the table $T_0$ otherwise to the table $T_1$. We say that a query element is part of the set if and only if the last query returns one.
\subsection{The Problem Statement} In this paper, we are dealing with the design of explicit adaptive scheme in the \emph{bitprobe} model to store an arbitrary subset of size at most five from a universe of size $m$ and answer the membership query in two adaptive bit probes. In other words our objective is to design an adaptive $(5,m,s,2)$-scheme in the \emph{bitprobe} model.
\subsection{Previous Results}
 As we are going to study a two adaptive \emph{bitprobe} scheme, let us discuss some existing results in the context of this problem. For the set of size one $(n=1)$, there exist a trivial scheme which takes $\mathcal{O}(m^{1/2})$  space. The space requirement for this scheme matches the lower bound of $\Omega (m^{1/2})$ \cite{buhrman2002bitvectors}. For the set of size two $(n=2)$, Radhakrishnan {\em et al.}~\cite{radhakrishnan2001explicit1} came up with a scheme which takes $\mathcal{O}(m^{2/3})$ space. Radhakrishnan et al. \cite{radhakrishnan2010data2} conjectured that this scheme is asymptotically tight but it has not been resolved yet. For the set of size three  $(n=3)$, Baig and Kesh \cite{baig2018two} came up with a scheme which takes $\mathcal{O}(m^{2/3})$ space. This scheme has been proved asymptotically tight by Kesh \cite{DBLP:conf/fsttcs/Kesh18}. \paragraph{}Moreover, for the set of size four $(n=4)$, Baig et al. \cite{DBLP:conf/iwoca/BaigKS19} have given a scheme which takes $\mathcal{O}(m^{5/6})$ space. This scheme improves upon the non-explicit $(n,m,c \cdot m^{1-\frac{1}{4n+1}},2)$-scheme   by Garg and Radhakrishnan \cite{garg2014set} for the set of size four $(n=4)$. For the given set size their scheme takes $\mathcal{O}(m^{16/17})$  space. Our scheme also improves upon the explicit $(n,m,c \cdot m^{1-\frac{1}{4n-1}},2)$-scheme given by Garg \cite{gthesis} for the set of size four $(n=4)$. His scheme takes $\mathcal{O}(m^{14/15})$ space for the given set size.  
 \paragraph{}
 In this paper, we have come up with a scheme for the set of size five $(n=5)$. Our scheme takes $\mathcal{O} (m^{10/11})$ space. This scheme improves upon the non-explicit scheme by Garg and Radhakrishnan \cite{garg2014set} for the set of size five $(n=5)$. For the set of size five $(n=5)$ their scheme takes $\mathcal{O}(m^{20/21})$  space. Our scheme also improves upon the explicit scheme given by Garg \cite{gthesis} for the set of size five $(n=5)$. His scheme takes $\mathcal{O}(m^{18/19})$ space for the set of size $(n=5)$. 

\section{The Approach to the Problem}
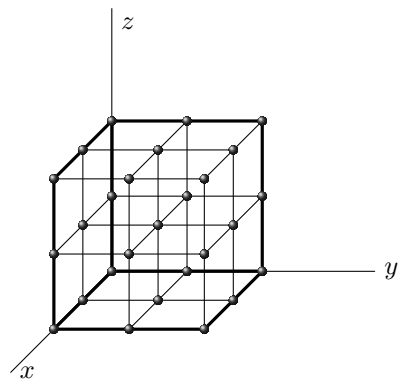
\begin{wrapfigure}{r}{0pt}
	\centering
	\vspace{-50pt}
	\begin{tikzpicture}[scale=1,
	cube/.style={very thick,black},
	grid/.style={very thin,gray},
	axis/.style={->,black,ultra thick}]
	\draw (0,0,0) -- (3.5,0,0) node[anchor=west]{$y$};
	\draw (0,0,0) -- (0,3.5,0) node[anchor=north west]{$z$};
	\draw (0,0,0) -- (0,0,3.5) node[anchor=west]{$x$};
	
	\draw[cube,fill=white!5] (0,0,0) -- (0,2,0) -- (2,2,0) -- (2,0,0) -- cycle;
	\draw[cube,fill=white!5] (0,0,0) -- (0,2,0) -- (0,2,2) -- (0,0,2) -- cycle;
	\draw[cube,fill=white!5] (0,0,0) -- (2,0,0) -- (2,0,2) -- (0,0,2) -- cycle;
	\foreach \x in {0,1,2}
	\foreach \y in {0,1,2}
	\foreach \z in {0,1,2}{
		
		\ifthenelse{  \lengthtest{\x pt < 2pt}  }
		{
			\draw [black]   (\x,\y,\z) -- (\x+1,\y,\z);
		}
		{
		}
		
		\ifthenelse{  \lengthtest{\y pt < 2pt}  }
		{
			\draw [black]   (\x,\y,\z) -- (\x,\y+1,\z);
		}
		{
		}
		\ifthenelse{  \lengthtest{\z pt < 2pt}  }
		{
			\draw [black]   (\x,\y,\z) -- (\x,\y,\z+1);
		}
		{
		}
		\shade[ axis,ball color = black!80] (\x,\y,\z) circle (0.06cm); 
		
	}
	
	\end{tikzpicture}
	
	\caption{Blocks of  Superblocks placed on the integral points of a Cube }
	\vspace{-30pt}	
	\label{fig:arrange5-2}
\end{wrapfigure}
\paragraph{}

Our scheme has borrowed the idea of the geometric arrangement of elements on a three-dimensional cube from Kesh~\cite{kesh2017adaptive}. Kesh in his paper used the idea of geometric arrangements of elements on high dimensional cubes to come up with $(2,m,c \cdot m^{1/(t-2^{-1})},t)$-scheme for $t \ge 2$. We have also used the idea of dividing the universe into blocks and superblocks from Radhakrishnan {\em et al.}~\cite{radhakrishnan2001explicit1}. Baig and Kesh~\cite{baig2018two} used the combination of the ideas mentioned above to come up with a tight explicit adaptive scheme for $n=3$. Baig {\em et al.}~\cite{DBLP:conf/iwoca/BaigKS19} used a similar idea to map elements on a square grid to come up with an improved scheme for $n=4$.  In this section, we use this geometrical technique to map the blocks of elements from superblocks to the integral point of a three-dimensional cube of side length $x$, as shown in Figure~\ref{fig:arrange5-2}. 

\paragraph{}
We divide the universe $\mathcal{U}$ of size $m$ into blocks and superblocks similar to the Radhakrishnan {\em et al.}~\cite{radhakrishnan2001explicit1}. We divide the universe into blocks of size $y$, so the number of blocks will be $m/y$. We then collect $x^{3}$ consecutive blocks to form a superblock of size $x^{3}y$. So we will have $m/x^{3}y$ superblocks of size $x^{3}y$.\\  
\\			
\textbf{Table {$T$}}\\ 
This table consists of one bit of space for each block. So the size of Table $T$ is $m/y$ bits.\\
\\
\textbf{Table {$T_1$}}
\\
Table $T_1$ is arranged in a three-dimensional cube of side $x$. So in this cube, we have $x^{3}$ integral points. Each integral point on or inside the cube contains a block of size $y$. So the size of table $T_1$ is $x^{3}y$. Since the size of each superblock is $x^{3}y$, all the blocks belonging to a superblock can be mapped on or inside the integral point of the cube. All other superblocks can be thought of as superimposed over each other in the cube. So each point in the cube or table $T_1$ is shared by blocks of $m/x^{3}y$ superblocks.\\
\\
\textbf{Table {$T_0$}}\\
While discussing the structure of table $T_1$, we saw that each superblock is mapped on a three-dimensional cube in such a way that all of them are superimposed. Now, for the $n$\textsuperscript{th} superblock, we first draw a family of lines in the bottom-most layer of the cube in the $XY$-plane with slope $1/n$ in such a way that all the integral points are covered by the lines. 
 \begin{figure}[t]
 	\centering
 	\captionsetup{justification=raggedright}  
 	\begin{minipage}{.45\textwidth}
 		\centering
 		\vspace{-.8cm}
 		\begin{tikzpicture}[scale=0.42]
 		\draw[step=1cm,gray,very thin] (0,0) grid (6,6);
 		\draw (0,0) -- (6,1.2) node[anchor=north west]{L}; 			
 		\draw [] (0,-1)-- (5,-1);
 		\draw [] (-1,0)-- (-1,6);
 		\draw [] (7,0)-- (7,1.2);
 		\draw [] (5.15,0)-- (5.15,1);
 		\draw (2.5,-1.5) node[anchor=north west] {n};
 		\draw (-1.75,3) node[anchor=north west] {x};
 		\draw (0,0) node[anchor=north west] {o};
 		\draw (7,0.75) node[anchor=north west] {y};
 		\draw (3,1.5) node[anchor=north west] {w};
 		\draw (5.3,0.75) node[anchor=north west] {1};
 		\end{tikzpicture}
 		\caption{A line with slope $1/n$ in the bottom most layer of the cube} 
 		\label{fig:lineslope}
 	\end{minipage}%
 \hspace{.2cm}
 	\begin{minipage}{.45\textwidth}
 		\centering
 		
 		\begin{tikzpicture}[scale=0.45]
 		
 		\draw[step=1cm,gray,very thin] (0,0) grid (6,6);
 		\draw (0,0) -- (5,1) node[anchor=north]{L};
 		\draw (0,1) -- (5,2) node[anchor=south]{M};
 		\draw (0,0.8) -- (1,1) node[anchor=north west]{};
 		\draw (0,0.6) -- (2,1) node[anchor=north west]{};
 		\draw (0,0.4) -- (3,1) node[anchor=north west]{};
 		\draw (0,0.2) -- (4,1) node[anchor=north west]{};
 		\draw (0,0) -- (6,0) node[anchor=north west]{X};
 		\draw (0,0) -- (0,6) node[anchor=north east]{Y};
 		\end{tikzpicture}
 		\caption{Figure showing number of lines drawn between two same slope lines} 
 		\label{fig:numlineslope}
 	\end{minipage}
 \end{figure}

\begin{lemma}
	\label{linescount}
	The number of lines passing through all the integral points of a square grid with slope $1/n$ is $2x+(n-1)(x-1)-1$, where $x$ is the length of the square grid.
\end{lemma}
\begin{proof}
	\vspace{-.23cm}
	As shown in Figure~\ref{fig:numlineslope}, if the slope of the line $M$ and $L$ is ${1}/{n}$, then between them, there can be only $n-1$ lines of slope $1/n$ passing through integral points of the grid. So the total number of lines that we can draw with slope $1/n$ is $x$ lines from integral points on $X$-axis, $x-1$ lines from the integral points on $Y$ axis and $(n-1)(x-1)$ lines between lines from the integral points on $Y$-axis. So we have $2x+(n-1)(x-1)-1$ lines with slope $1/n$. 
\end{proof}
  
\begin{wrapfigure}{R}{0.4\textwidth}
	\begin{center}
		\vspace{-20pt}
		\includegraphics[width=0.4\textwidth]{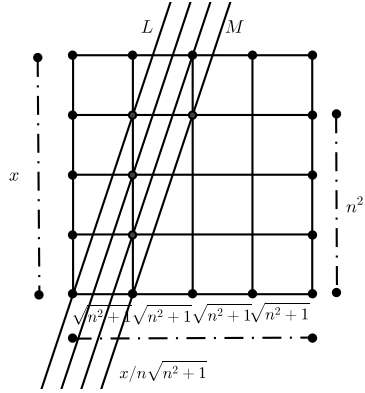}
		\vspace{-25pt}
	\end{center}
	\caption{A Slice Belonging to $n$\textsuperscript{th} Superblock}
	\label{slice}
\end{wrapfigure}
Using Lemma~\ref{linescount}, we can say  that the total number of lines with slope $1/n$ in the bottom-most layer can be $c \cdot nx$,  
where $c$ is a constant, and $x$ is the side of the cube. Now, we cut slices of the cube along these lines and perpendicular to $XY$-plane. So the total number of slices for the $n$\textsuperscript{th} superblock will be equal to the number of lines drawn in the bottom-most layer of the cube for the $n$\textsuperscript{th} superblock, i.e $c \cdot nx$. All the slices have a height equal to the length of the cube, i.e $x$. Let us now calculate the maximum width of a slice of slope $1/n$. Width of the slice formed by line segment $OL$ as shown in Figure~\ref{fig:lineslope} can be calculated to be $ {x}/{n} \sqrt{1+n^{2}}$. We can see from  Figure~\ref{fig:lineslope}  that all other slices with this slope will have width less than or equal to $ {x}/{n} \sqrt{1+n^{2}}$. So a slice belonging to $n$\textsuperscript{th} superblock will have length $x$ and width less than equal to $ {x}/{n} \sqrt{1+n^{2}}$. 
\newline We now draw a family of lines in all the slices of all the superblocks. For the slices belonging to the $n$\textsuperscript{th} superblock, we draw lines with slope
${n^{2}}/{\sqrt{1+n^{2}}}$ as shown in Figure~\ref{slice}. The lines are drawn in such a way that all the integral points are covered. Let us now calculate the maximum number of lines drawn on a slice belonging to the $n$\textsuperscript{th} superblock. We can see from Figure~\ref{slice} that the total number of lines drawn from $Z$-axis is equal to $x$. Also, number of integral points on the width of the slice is equal to $ {x}/{n} \sqrt{1+n^{2}}/\sqrt{1+n^{2}}= {x}/{n}$. So we can draw $\frac{x}{n}$  lines through those integral points on the width of the slice. Now from Figure~\ref{slice}, we can see that the number of lines that can pass between two consecutive integral points on the width of the slice is $n^{2}-1$. So the total number of lines drawn on this slice with slope ${n^{2}}/{\sqrt{1+n^{2}}}$ is equal to $x+ {x}/{n}+ n^{2}\cdot {x}/{n}$ i.e $c \cdot nx$. We say that a slice is having slope $1/n$ if it's projection on the $XY$-plane has slope $1/n$. Now let us bound the total number of lines drawn on slices whose projections on the $XY$-plane makes slope $1/n$. The total number of lines should be less than the product of the number of lines drawn on a slice of a maximum width of slope $1/n$ and the total number of slices of slope $1/n$. So the total number of lines drawn for the $n$\textsuperscript{th} superblock is less than $c_{1}\cdot n^{2}x^{2}$. We need to sum this for all the superblocks to get the total number of lines drawn. So the total number of lines drawn is  
\begin{equation}
\sum_{i=1}^{m/x^{3}y} c_{1}\cdot i^{2} (x)^{2} \le c \cdot \frac{m^{3}}{x^{7}y^{3}}.
\end{equation}
For each line drawn in a slice, we have a block of space in table $T_0$. So the total size of table $T_0$ is $c \cdot {m^{3}}/{x^{7}y^{2}}$ bits. Summing up the space taken by tables $T, T_0$ and $T_1$, we get the following equation:

\begin{equation}
S(x,y)= x^{3}y+C \cdot (\frac{m^{3}}{x^{7}y^{2}})+\frac{m}{y} .
\end{equation}

Choosing $x= m^{3/11}$ and $y= m^{1/11}$, we get space taken by our data structure to be $\mathcal{O}({m^{10/11}})$. 

Now we will prove the following lemma:

\begin{lemma}
	No three lines passing through an integral point of the cube lies in the same plane.
\end{lemma}
\begin{proof}
	From the construction of the table, we can see that all the lines drawn in the cube for a given superblock are parallel to each other. So the lines which pass through the same integral point of the cube belong to the different superblocks. Let us consider the three arbitrary superblocks to which our lines belong. Without loss of generality let us say that the projection of these slices on the $XY$-plane makes angle  $ {1}/{n_1}, {1}/{n_2}$ and ${1}/{n_3}$ with the $X$-axis. Our lines lie completely in the slices belonging to their superblock. While drawing lines in the slices for the first, second, and third superblock, we are going up in $Z$ direction by $n_1^{2}, n_2^{2}$, and $n_3^{2}$. Hence our lines cannot lie in the same plane.
	
\end{proof}

\section{The Adaptive Scheme for Five Elements}
In this section we will present our $(5,m,\mathcal{O} (m^{10/11}),2)$-scheme.
\subsection{Our Data Structure}
Our data structure has three tables $T,T_0$, and $T_1$. Structure of these tables have been discussed earlier, and we have seen that size of each table is $\mathcal{O} (m^{10/11})$. So the size of our data structure is $\mathcal{O} (m^{10/11})$.
\paragraph{}
 From the structure of tables, we may draw the following conclusion. In general, two blocks having elements should not map at the same location in table $T_1$ or $T_0$. Otherwise, we may make a mistake on the query belonging to these blocks. Further, if the block having element is mapped in table $T_1$ or $T_0$, then no other block should be sent to that table whose position is matched with the block having an element. So if a block having an element on a line is mapped to table $T_0$, then all other blocks lying on that line should be sent to table $T_1$. As for each line, we have only one block of space in table $T_0$.
 On the contrary, if the block having an element from a line is sent to table $T_1$, then other blocks lying on the line which contains this block can be sent to table $T_0$ or $T_1$. The blocks which are not having any elements given to be stored can be mapped at the same location in table $T_1$ or $T_0$. Further, in the rest of the paper, whenever we say the line passing through a block or the line having a block, we always mean the line drawn in the superblock to which the block belongs.  

\subsection{The Query Scheme}
 Given a query element, we find the block and superblock to which it belongs. Further, we query the first table, if the first table returns zero, we query to table $T_0$ else we query to table $T_1$. We say that element is part of the set to be stored if and only if the last query returns one.

\subsection{The Storage Scheme} 
\label{bib1}
In this section, we talk about the way bits of tables are set to store the subset of size at most five from a universe of size $m$. We divide the storage scheme into various cases depending upon the way blocks having elements are distributed on the line belonging to their superblock. To generate all the cases, first of all, we partition the number five; then we put those many elements into different superblocks. Further, the positions of the blocks having elements on the line belonging to their superblocks are considered. While handling cases, we see the intersections of the lines, which contains blocks having elements given to be stored. We then decide which block to send to table $T_0$ and which to $T_1$. As in our data structure, we always send a block to either table $T_0$ or $T_1$, and we store its bit vector there, so we will always assume that elements which are given to be stored lies in the different block. Proving the results for elements belonging to different blocks proves the result when many elements belong to the same block. In this section, we will discuss few cases. The rest of the cases can be generated and handled similarly, and are mentioned in~\ref{bib}. \\
\\
\textbf{Case 1.}
If all the elements of $S$ lie in one superblock, then we send the blocks having elements to table $T_1$ and all the empty blocks to table $T_0$.\\
\\     
\textbf{Case 2.}
If four elements $S_1=\{n_1,n_2,n_3,n_4\}$ lie in one superblock and one element $S_2=\{n_5\}$ in other superblock then we can have two cases, either the block containing the element $n_5$ coincides with one of the block containing element from $S_1$ in table $T_1$ or it does not coincides. So if the block containing the element $n_5$ coincides with one of the blocks containing an element from $S_1$, then we send the block having the element $n_5$ to table $T_1$ and send the block from which it was coinciding to table $T_0$. All other blocks of superblock which contain elements from  $S_1$ are sent to table $T_1$. All other blocks of superblock which contain the element $n_5$  are sent to table $T_0$. Rest all the empty blocks are sent to table $T_0$. On the other hand, if the block containing element $n_5$ do not coincide with any of the block having element from $S_1$ in table $T_1$, then we send all the blocks having elements from $S_1$ and $S_2$ to table $T_1$, and rest all the empty blocks to table $T_0$.\\ 
\\
\textbf{Case 3.}
If three elements $ S_1= \{n_1,n_2,n_3\}$ lie in one superblock and two elements $S_2=\{n_4,n_5\}$ in other superblock then we store according to following scheme.\\
\\
\textbf{Case 3.1}
All the blocks to which elements from $S_1$ belong lies on the same line of their superblock. From here onwards, whenever we say line passing through a block or blocks lying on a line, we mean the line drawn in the superblock to which these blocks belong.\\
\\        
\textbf{Case 3.1.1}        
Two blocks to which elements from $S_2$ belong coincides with the blocks corresponding to the elements from $S_1$ in table $T_1$. In this case, we send the blocks having elements from $S_2$ to table $T_0$. Further, we send empty blocks lying on the lines to which elements from $S_2$ belongs to table $T_1$. We send all the blocks which contain elements from $S_1$ in table $T_1$. Finally, we send the rest of the empty blocks to table $T_0$.\\  
\\                 
\textbf{Case 3.1.2}

Only one block which contains an element from $S_2$  coincides with the block corresponding to the elements from $S_1$ in table $T_1$. In this case, we send all the blocks which contain elements from $S_1$ to table $T_1$. We send the coinciding block of the element from $S_2$ to table $T_0$ and the rest of the blocks, which lies on the line containing this block to table $T_1$. If after this other nonempty block having an element from $S_2$ is still unassigned, then we send it to table $T_1$, and all the empty blocks lying on the line containing this block to table $T_0$ . Rest all the empty blocks are sent to table $T_0$ .\\   
\\        
\textbf{Case 3.1.3}
None of the blocks which contain an element from $S_2$ coincides with the block, which contains an element from $S_1$ in table $T_1$ . In this case, we send all the nonempty blocks to table $T_1$ and all the empty blocks to table $T_0$. \\  
\\            
\textbf{Case 3.2}
Two blocks that contain elements from $S_1$ lies on the same line,  and another lie on a different line.\\
\\
\textbf{Case 3.2.1}
All the blocks which contain an element from $S_1$ lies in the same slice. From here onward, whenever we say blocks belonging to a slice, we mean the slice drawn in a superblock to which these blocks belong.\\
\\             
\textbf{Case 3.2.1.1}
All the blocks which contain elements from  $S_2$ coincides with blocks which contain elements from $S_1$ in table $T_1$. In this case, we can use the assignment made in Case 3.1.1. \\       
\\            
\textbf{Case 3.2.1.2} Only one block which contains an element from $S_2$ coincides with the block, which contains an element from $S_1$ in table $T_1$. In this case, we can use the assignment made in Case 3.1.2\\   
\\            
\textbf{Case 3.2.1.3}
None of the blocks which contain an element from $S_2$ coincides with the block, which contains an element from $S_1$ in table $T_1$. This case is the same as Case 3.1.3.\\ 
\\                    
\textbf{Case 3.2.2} Two blocks that contain elements from $S_1$ lies in a slice and another block that contains an element from $S_1$ in another slice.\\
\\
\textbf{Case 3.2.2.1} Both the blocks which contain elements from  $S_2$ coincides with blocks which contain elements from $S_1$ in table $T_1$. If the coinciding blocks lie in the same slices as that of two blocks that contain elements from $S_1$, then we send both the coinciding blocks which contain the elements from $S_2$ to table $T_0$. Further, we send all the empty blocks lying on the lines which contain these blocks to table $T_1$. Two blocks that contain the elements from $S_1$ and lying in the same slice are sent to table $T_1$. The remaining block which contains the element is sent to table $T_0$ and all the empty blocks lying on the line containing this block to table $T_1$. Rest all the empty blocks are sent to table $T_0$.\newline Now consider the case in which the coinciding blocks which contain the elements from $S_2$ lies in different slices of the superblock, which contains the elements from $S_1$. In this case, we send the two blocks which contain the elements from $S_1$ lying in a slice to table $T_1$. The rest of the block, which contains the element from $S_1$ lying in the other slice, is sent to table $T_0$, and the empty blocks which lie on the line containing this block are sent to table $T_1$. One of the blocks which contains an element from $S_2$ and is lying in the slice containing two elements from $S_1$ is sent to table $T_0$  and the rest of the blocks on this line is sent to table $T_1$. If after this assignment, other block having the element from $S_2$ is still unassigned, then we send it to table $T_1$, and the empty blocks on the line containing this block to table $T_0$. Rest all the empty blocks are sent to table $T_0$.\\   
\\
\textbf{Case 3.2.2.2}
Only one block which contains an element from $S_2$ coincides with a block that contains an element from $S_1$ in table $T_1$. Let us first consider the case where coinciding block having element form $S_2$ lies in the slice, which contains two blocks having elements from $S_1$. Without loss of generality, let us say that the blocks having elements $n_1$ and $n_2$ lies in the same slice and the block having the element $n_4$ coincide with the block having the element $n_1$. In this case, we send the block having the element $n_4$ to table $T_0$, and all the blocks on the line containing this block is sent to table $T_1$.\newline 
If the block which contains an element $n_3$ lies on the line which contains the block having the element $n_4$, then we send the block having $n_3$ to table $T_0$, and empty blocks of the line which contains block having $n_3$ to table $T_1$. Now, if the block having the element $n_5$ is still unassigned, then we send the block having the element $n_5$ to table $T_0$, and empty blocks on the line containing this block to table $T_1$. We send the block having the element $n_2$ to table $T_0$, and all the blocks which lie on the line containing this block table $T_1$. Rest all the empty blocks are sent to table $T_0$.\newline Further, let us consider the case where the block that contains the element $n_3$ does not lie on the line, which contains the block having the element $n_4$. In this case, we send the block having $n_3$ to table $T_1$, and the empty blocks on the line containing this block are sent to table $T_0$. Now, if the block having the element $n_5$ is still unassigned, then we send it to table $T_0$, and the empty blocks lying on the line containing this block to table $T_1$. We send the block having the element $n_2$ to table $T_0$, and all the blocks which lie on the line containing this block table $T_1$. Rest all the empty blocks are sent to table $T_0$.\newline
Now we are left with the case where coinciding block of $S_2$ having an element $n_4$ coincides with the block having the element $n_3$. In this case, we send the block having the element $n_3$ to table $T_0$ and empty block, which lies on the line containing this block to table $T_1$. We send the block having the element $n_4$ to table $T_1$ and empty blocks which lies on the line containing this block to table $T_0$.\newline Now we see the position of the block having the element $n_5$. If the block having the element $n_5$ lies on the line, which contains block having the element $n_3$, then we send the block having the element $n_5$ to table $T_0$.  Further, we send the empty blocks of the line, which contains block having the element $n_5$ to table $T_1$. Now, consider the case where the line containing the block having the element $n_5$ passes through one of the blocks having the element $n_1$ or $n_2$. Without loss of generality, let us say that the line containing the block having the element $n_5$ passes through the block having the element $n_1$. In this case, we send the block having the element $n_1$ to table $T_0$ and rest all the blocks lying on this line to table $T_1$. Rest all the empty blocks are sent to table $T_0$. If the line which contains the block having element $n_5$ does not pass through the block block having element $n_1$ or $n_2$, in this case we can send both the blocks having elements $n_1$ and $n_2$ to table $T_1$. Rest all the empty blocks are sent to table $T_0$.\newline Now consider the case where block having the element $n_5$ does not lie on the line, which contains block having the element $n_3$. In this case, we send the block having the element $n_5$ to table $T_1$ and all the empty blocks lying on the line containing this block to table $T_0$. Blocks having elements $n_1$ and $n_2$ are sent to table $T_1$, and rest all the empty blocks are sent to table $T_0$.\\
\\
\textbf{Case 3.2.2.3}
None of the blocks which contain elements from $S_2$ coincide with blocks having elements from $S_1$ in table $T_1$. This case is the same as Case 3.1.3. \\
\\                
\textbf{Case 3.3}
All the blocks which contain elements from $S_1$ lies on the different lines.\\
\\
\textbf{Case 3.3.1}
Both the blocks having elements $n_4$ and $n_5$ coincides with the blocks having elements from $S_1$ in table $T_1$. Without loss of generality, let us say that block having element $n_1$, coincides with the block having the element $n_4$, and the block having the element $n_2$ coincide with the block having the element $n_5$. In this case, we send the blocks having elements $n_1,n_2$ and $n_3$ to table $T_0$ and all the empty blocks lying on the lines which contain these blocks to table $T_1$. Further, we send the blocks having the element $n_4$ and $n_5$ to table $T_1$. Rest all the empty blocks are sent to table  $T_0$.  \\
\\    
\textbf{Case 3.3.2}
Only one of the block having element say $n_4$ from $S_2$ coincides with blocks having element from $S_1$ in table $T_1$. Without loss of generality, let us say that block having the element $n_1$ coincides with the block having the element $n_4$. Similar to the last case, in this case also, we send the blocks having elements $n_1,n_2$ and $n_3$ to table $T_0$, and all the empty blocks lying on the lines which contain these blocks to table $T_1$. Further, We send the block having element $n_4$ to table $T_1$. We send the block having element $n_5$ to table $T_0$, and all the empty blocks lying on the line containing this block to table $T_1$. Rest all the empty blocks are sent to table $T_0$.\\       
\\        
\textbf{Case 3.3.3}
None of the blocks which contain an element from $S_2$ coincides with the block, which contains an element from $S_1$. This case is the same as Case 3.1.3.\\

\paragraph{\textbf{Correctness}} 
The correctness of the scheme relies on the fact that blocks having the elements do not coincide in table $T_1$ or in table $T_0$. Also, the blocks which are not having the elements are not sent to the place where block having elements are mapped.\\
\\
We summaries the conclusion of this section as follows.

\begin{theorem}
	\label{thm:final3}
	There is a two probe explicit adaptive scheme which stores an arbitrary subset $\mathcal{S}$ of size at most five from a universe $\mathcal{U}$ of size \emph{m} and uses $\mathcal {O}(m^{10/11})$ bits of space.
\end{theorem}

\subsection{Counterexample}

\begin{wrapfigure}{R}{0.4\textwidth}
	\begin{center}
		\vspace{-10pt}
		\includegraphics[width=0.45\textwidth]{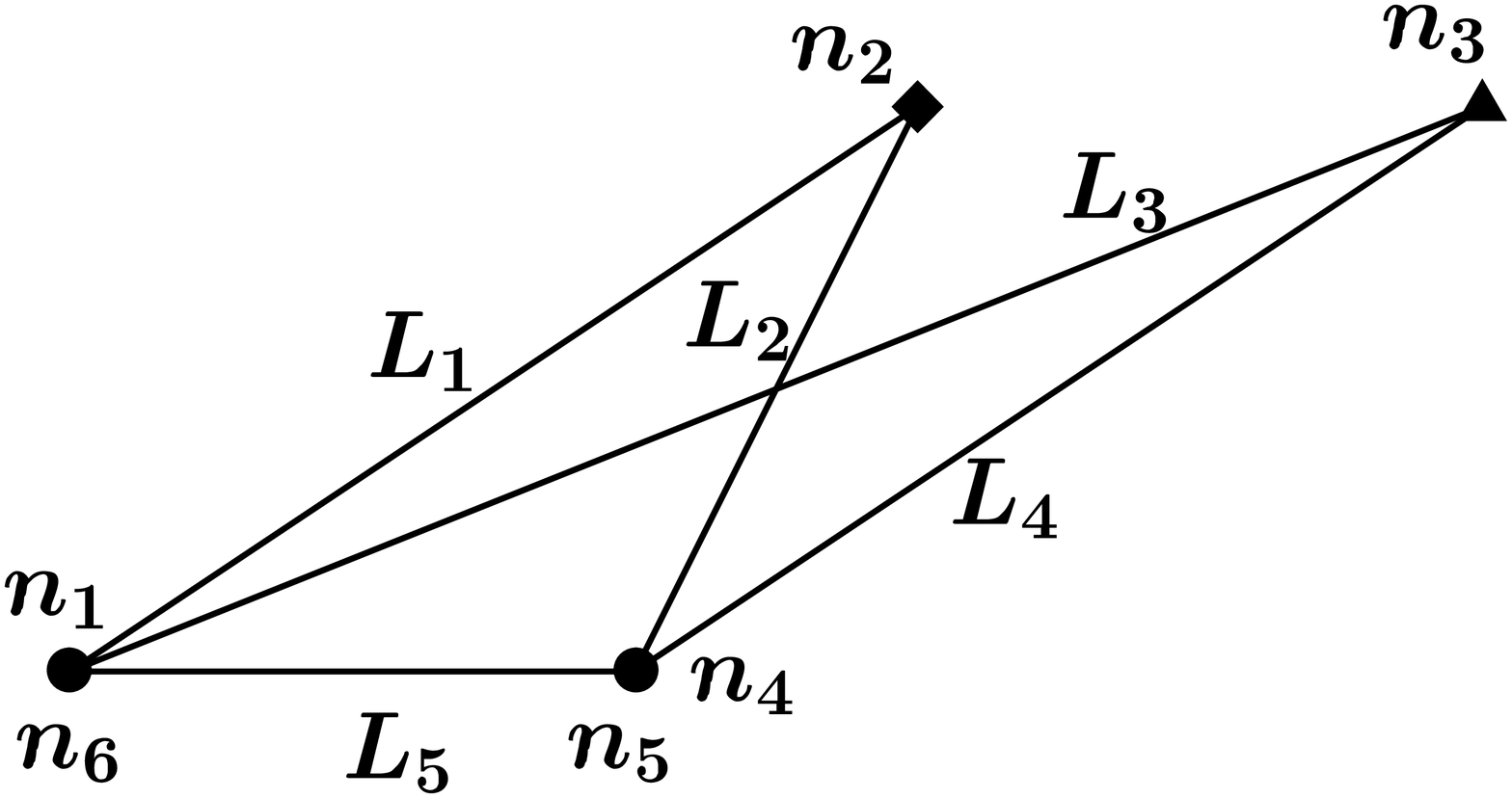}
		\vspace{-25pt}
	\end{center}
	\caption{A counterexample for a six elements subset}
	\vspace{-20pt}
	\label{count:final8}
\end{wrapfigure}
In this subsection, we will show that the above-mentioned scheme can not store a subset $\mathcal{S}$ of size six. Let us consider a subset  $\mathcal{S} =\{n_1,n_2,n_3,n_4,n_5,n_6\}$ of six elements from a universe $\mathcal{U}$ of size $m$. Further, let us consider that all these elements belong to the different blocks. As shown in the Figure~\ref{count:final8}, blocks having the elements $n_1,n_2,n_3$ and $n_4$ lies on the line $L_1,L_2,L_3$ and $L_4$ respectively. Whereas, blocks having the elements $n_5$ and $n_6$ lies on the line $L_5$. Now, we claim that if we store the configuration shown in the Figure~\ref{count:final8} in our {$(5,m,m^{10/11},2)$}-scheme then our query scheme will answer incorrectly.\newline To show that our claim is true, let us consider the block having the element $n_1$. Now, this block can either go to table $B$ or to table $C$. Let us first consider the case where the block having the element $n_1$ goes to table $B$. 
\newline If the block having the element $n_1$ goes to table $B$, then rest all the blocks lying on the line $L_1$ must go to table $C$. Therefore, the block having the element $n_2$ must go to table $B$. Now, rest all the blocks lying on the line $L_2$ must go to table $C$. Therefore, blocks having the elements $n_4$ and $n_5$ must go to table $B$. If the block having the element $n_4$ goes to table $B$, then rest all the blocks lying on the line $L_4$ must go to table $C$. Therefore, the block having the element $n_3$ must go to table $B$. Now, rest all the blocks lying on the line $L_3$ must go to table $C$. Therefore, block having the element $n_6$ must go to table $B$. Since we have already sent the block having the element $n_5$ to table $B$, so our query scheme will answer incorrectly for the queries belonging to these blocks.\newline Now, we are left with the case where the block having the element $n_1$ goes to table $C$. In this case, the block having the element $n_6$ must go to table $B$. Therefore, rest all the blocks lying on the line $L_5$ must go to table $C$. Now, the block having the element $n_4$ must go to table $B$. Therefore, rest all the blocks lying on the line $L_4$ must go to table $C$. It forces the block having the element $n_3$ to table $B$. Therefore, rest all the blocks lying on the line $L_3$ must go to table $C$. However, we have already sent the block having the element $n_1$ to table $C$. So our query scheme will answer incorrectly for the queries belonging to the block having the element $n_1$ and the empty block lying on the line $L_3$ and coinciding with the block having the element $n_1$.
\newline The block containing the element $n_1$ can either go to table $B$ or to table $C$, and our query scheme answers incorrectly in both the cases. Therefore, the above configuration of the elements cannot be stored in the $(5,m,m^{10/11},2)$-scheme.

\section{Conclusion}
In this paper, we have come up with an explicit adaptive $(5,m,\mathcal{O} (m^{10/11}),2)$-scheme, which improves upon the non-explicit scheme by Garg and Radhakrishnan~\cite{garg2014set} and the explicit scheme by Garg~\cite{gthesis} for the given set and query sizes. We have borrowed the idea of the geometrical arrangement of elements on the three-dimensional cube from Kesh \cite{kesh2017adaptive} and the idea of dividing the universe into blocks and superblocks from Radhakrishnan {\em et al.}~\cite{radhakrishnan2001explicit1}. Using these ideas there are improved schemes for the set of size three, four and five. We believe that this idea can be further generalized to improve the existing results for arbitrary subsets of size $n$.		

\bibliographystyle{splncs03}
\bibliography{paper}

\newpage
\appendix
\renewcommand\thesection{Appendix \Alph{section}}
\section{}
\label{bib}
In this section, we discuss rest of the cases of Section ~\ref{bib1}.\\
\\
\textbf{Case 4}
The elements in $S_1 =\{n_1,n_2,n_3\}$  lies in a superblock and the elements $n_4$ and $n_5$ in the different superblocks.\\ 
\\               
\textbf{Case 4.1}
Elements $n_1,n_2$ and $n_3$ lies on a same line.\\
\\                   
\textbf{Case 4.1.1}
Both the elements $n_4$ and $n_5$ coincides with the blocks having elements from $S_1$ in table $T_1$. In this case, we send the blocks having elements $n_4$ and $n_5$ to table $T_0$, and the rest of the empty blocks which lie on the lines containing these blocks to table $T_1$. We send the blocks having elements from $S_1$ to table $T_1$. Rest all the empty blocks are sent to the table $T_0$.\\   
\\                   
\textbf{Case 4.1.2}
Only one of the blocks having an element say $n_4$ coincides with the block having an element from $S_1$ in table $T_1$. Without loss of generality, let us say that block having the element $n_1$ coincides with the block having the element $n_4$. In this case, we send the blocks having elements $n_4$ and $n_5$ to table $T_0$ and all the empty blocks lying on the line containing these blocks to table $T_1$. Now we see whether the line, which contains block having the element $n_5$ passes through the block having element elements from $S_1$ or not. Let us first consider a case where the line which contains block having the element $n_5$ passes through one of the blocks having elements from $S_1$, without loss of generality, let us say it passes through block having the element $n_2$. In this case, we send the block having the element $n_2$ to table $T_0$, and rest all the blocks lying on the line containing this block to table $T_1$. Rest all the empty blocks are sent to table $T_0$. On another hand, if the line which contains block having the element $n_5$ does not pass through any of the block having element from $S_1$, then we send the blocks having elements from $S_1$ to table $T_1$. Rest all the empty blocks are sent to table $T_0$. \\ 
\\
\textbf{Case 4.1.3}
None of the blocks which contains element  $n_4$ or $n_5$ coincides with the blocks which contains an element from $S_1$ in table $T_1$. This case is the same as Case 3.1.3.\\
\\                   
\textbf{Case 4.2}
Two elements say $n_1$ and $n_2$ lies on a same line, and the element $n_3$ lies on a different line.\\
\\                       
\textbf{Case 4.2.1}
Blocks having the elements $n_4$ and $n_5$ coincides with the blocks having elements from $S_1$ in table $T_1$.\\
\\
\textbf{Case 4.2.1.1}
Blocks having the elements $n_4$ and $n_5$ coincides with the block having elements $n_1$ and $n_2$ in table $T_1$. In this case, we send the blocks having elements $n_4$ and $n_5$ to table $T_0$, and all the empty blocks lying on the lines containing these blocks to table $T_1$. Further, we send the blocks having elements $n_1$ and $n_2$ to table $T_1$. We send the block that contains the element $n_3$ to table $T_0$, and the rest of the empty block lying on the line containing this block to table $T_1$. Rest all the empty blocks are sent to the table $T_0$.\\  
\\                           
\textbf{Case 4.2.1.2}
Blocks having elements $n_4$ and $n_5$ coincides with the block lying on the different line, say $n_1$ and $n_3$ in table $T_1$. In this case, we send the blocks having elements $n_4$ and $n_5$ to table $T_1$ and the rest of the empty blocks lying on the line containing these blocks to  table $T_0$ .  We send the blocks having elements $n_1$ and $n_3$ to the table $T_0$ and the rest of the blocks lying on these lines to the table  $T_1$. Rest all the empty blocks are sent to table $T_0$.\\  
\\                           
\textbf{Case 4.2.1.3}
Blocks having element $n_4$ and $n_5$ coincides in table $T_1$. If it coincides with the block having the element $n_1$ or $n_2$, then we send both the block having the element $n_4$ and $n_5$ to table $T_0$, and the rest of the empty block lying on the line containing these blocks to table $T_1$. Further, we send the blocks having the element $n_1$ and $n_2$ to table $T_1$. We send the block that contains element $n_3$ to the table $T_0$, and the rest of the empty block, which lies on the line containing this block to the table $T_1$. Rest all the empty blocks are sent to the table $T_0$.\newline If the blocks containing elements $n_4$ and $n_5$ coincides with the block which contains element $n_3$ in table $T_1$, then we send the block containing element $n_3$ to the table $T_0$, and the rest of the empty block, which lie on the line containing this block to the table $T_1$. Further, we send the block containing $n_4$ to table $T_1$, and the rest of the empty block lying on the line containing this block to the table $T_0$. We send the block having the element $n_5$ to the table $T_0$, and the rest of the empty block lying on the line containing this block to table $T_1$. Now we see whether the line which contains block having the element $n_5$ passes through the block having the element  $n_1$ or $n_2$. Without loss of generality, let us first consider the case where the line which contains block having the element $n_5$ passes through the block having the element $n_1$. In this case, we send the block having the element $n_1$ to table $T_0$, and all the blocks lying on the line containing this block to table $T_1$. Rest all the empty blocks are sent to table $T_0$. If the line which contains block having element $n_5$ does not pass through the blocks having elements $n_1$ or $n_2$, then we send the blocks having elements $n_1$ and $n_2$ to table $T_1$, and rest all the empty blocks to table $T_0$.\\ 
\\                           
\textbf{Case 4.2.2}
Only one block having an element say $n_4$ coincides with the block having the element from $S_1$ in table $T_1$. Let us first consider the case where block having element $n_4$ coincide with the block having element $n_1$ or $n_2$. Without loss of generality, let us say block having the element $n_4$ coincide with the block having the element $n_1$. In this case, we send the block having the element $n_4$ to table $T_1$, and all the empty blocks lying on the line containing this block to table $T_0$. We send the block having the element $n_1$ to table $T_0$, and rest all the blocks lying on the line containing this block to the table $T_1$. Now we see the position of the block having the element $n_5$. If it lies on the line which contains block having the element $n_1$, then we send the block having the element $n_5$ to table $T_0$, and all the empty blocks lying on the line containing this block to the table $T_1$. We send the block having the element $n_3$ to table $T_0$, and all the empty blocks lying on the line containing this block to table $T_1$. Rest all the empty blocks are sent to table $T_0$. On another hand, if the block having the element $n_5$ does not lie on the line, which contains block having the element $n_1$, then we send the block having the element $n_3$ and $n_5$ to table $T_1$. Rest all the empty blocks are then sent to table $T_1$.\newline Now let us consider the case where block having the element $n_4$ coincide with the block having the element $n_3$. In this case, we send the block having the element $n_4$ to table $T_1$, and rest all the empty blocks lying on the line containing this block to table $T_0$. We send the block having the element $n_3$ to table $T_0$, and all the empty blocks lying on the line containing this block to table $T_1$. Now we see the position of the block having the element $n_5$. If the block having the element $n_5$ lies on the line which contains block having the element $n_3$ , then we send the block having the element $n_5$ to the table $T_0$, and all the empty blocks lying on the line containing this block to the table $T_1$. If the block having the element $n_1$ lies on the line which contains block having the element $n_5$, then we send the block having the element $n_1$ to table $T_0$, and rest all the blocks lying on the line containing this block to table $T_1$. Rest all the empty blocks are sent to the table $T_1$. Similar is the case if the block having the element $n_2$ lies on the line, which contains block having the element $n_5$. If the block having the element $n_5$ does not lie on the line which contains block having the element $n_3$, then we send the blocks having the elements $n_1,n_2$ and $n_5$ to table $T_1$, and rest all the empty blocks to table $T_0$.\\  
\\                           
\textbf{Case 4.2.3}    
None of the blocks having the elements $n_4$ or $n_5$ coincides with blocks having elements from $S_1$ in table $T_1$. This case is the same as Case 3.1.3.\\    
\\                   
\textbf{Case 4.3}    
All the elements belonging to $S_1$ lies on the different lines.  In this case, we send all the blocks having elements to table $T_0$, and all the empty blocks to table $T_1$.\\
\\	   				
\textbf{Case 5}
Two elements $S_1=\{n_1,n_2\}$ lies in a superblock other two elements $S_2=\{n_3,n_4\}$ lies in other superblock, and an element $S_3=\{n_5\}$ in a different superblock.\\
\\    
\textbf{Case 5.1}
Elements belonging to $S_1$ lie on the same line, elements belonging to $S_2$ lies on a line.\\
\\        
\textbf{Case 5.1.1}
Two blocks having elements from $S_2$ and $S_3$ coincides with the blocks having elements from $S_1$ in table $T_1$. Without loss of generality, let us say that block having element $n_3$ coincides with the block having the element $n_1$, and the block having the element $n_5$ coincides with the block having the element $n_2$. In this case, we send the block having the element $n_4$ to table $T_0$, and the rest of the block lying on the line containing this block to table $T_1$. We send the block having the element $n_1$ to table $T_0$, and the rest of the block lying on this line to table $T_1$. Further, we send the block having the element $n_5$ to table $T_0$, and the rest of the empty block lying on the line containing this block to table $T_1$. Rest all the empty blocks are sent to table $T_0$.\newline If the blocks from $S_2$ and $S_3$ coincides with the same block having the element from $S_1$, then we send the coinciding block having elements from $S_2$ and $S_3$ to table $T_0$, and the rest of the blocks lying on these lines to table $T_1$. Furthermore, we send the blocks having elements from $S_1$ to table $T_1$. Rest all the empty blocks are sent to table $T_0$.\\ 
\\        
\textbf{Case 5.1.2}
Only one block having the element from $S_2$ or $S_3$ coincides with the block having the element from $S_1$.\\
\\            
\textbf{Case 5.1.2.1} 
Block having an element from $S_2$ coincides with the block having an element from $S_1$. Without loss of generality, let us say that the block having the element $n_3$ coincides with the block having the element $n_1$. Now, we can have two cases, either the block containing element $n_5$ coincides with the block containing element $n_4$ or it does not. \newline If the block containing $n_5$ coincides with the block containing $n_4$, then we send the block containing $n_4$ to table $T_0$, and the rest of the block lying on the line containing this block to the table $T_1$. We send the block containing $n_5$ to table $T_1$, and the rest of the block lying on the line containing this block to table $T_0$. Further, we send the block having the element $n_1$ to table $T_0$, and the rest of the block lying on this line to table $T_1$. Rest all the empty block are sent to table $T_0$.\newline If the block containing $n_5$ does not coincide with the block having the element $n_4$, then we send the block having the element $n_3$ to table $T_0$, and the rest of the block lying on this line to table $T_1$. Now we see the position of the block that contains the element $n_5$ to make the assignment. If the block which contains element $n_5$ lies on the line which contains block having elements from $S_1$, then we send the block having the element $n_5$ to table $T_1$, and rest all the empty blocks lying on the line which contains this block to table $T_0$.  Further, we send blocks having elements from $S_1$ to table $T_1$, and the rest of the empty blocks lying on the line which contains this block to table $T_0$. For the rest of the empty blocks, we send it to table $T_0$. \newline Rest for all other positions of block having element $n_5$, we send block having elements $n_3,n_4$ and $n_5$ to table $T_1$, and all the empty blocks lying on the lines containing these blocks to table $T_0$.  Further, we send the block having the element $n_1$ to table $T_0$, and all other blocks lying on the line containing this block to table $T_1$. Rest all the empty blocks are sent to table $T_0$. \\  
\\            
\textbf{Case 5.1.2.2}
Block having an element from $S_3$ coincides with the block having an element from $S_1$. Without loss of generality, let us say that block having element $n_5$ coincides with the block having the element $n_1$. In this case, we send the block having the element $n_5$ to table $T_0$, and the rest of the block lying on the line, which contains this block to table $T_1$. Now we see the position of the block having the element from $S_2$. \newline One of the blocks having an element from $S_2$ coincides with the empty block on the line, which contains block having the element $n_5$. Without loss of generality, let us say that block having element $n_3$ coincides with an empty block on the line, which contains block having the element $n_5$. In this case, we send the block having an element $n_3$ to table $T_0$, and the rest of the block lying on the line containing this block to table $T_1$. Now we see the position of the block having element $n_4$, if it lies on the line which contains block having element from $S_1$, then we send the block having element $n_1$ and $n_2$ to table $T_1$, and the rest of the empty block lying on the  line containing these blocks to table $T_0$. Rest all the empty blocks are sent to table $T_0$. On another hand, if the block having the element $n_4$ does not lie on the line which contains block having the element from $S_1$, then we send the block having the element $n_2$ to table $T_0$, and the rest of the block lying on the line containing  this block to table $T_1$. We send the rest of the empty block to table $T_0$.\newline If the block having the element from $S_2$ do not lie on the line which contains block having the element $n_5$, then we send the blocks having an element from $S_2$ to table $T_1$, and the rest of the empty blocks on the line which contains these blocks to table $T_0$. Also, we send the blocks having the element $n_1$ and $n_2$ to table $T_1$. We send rest of the empty blocks to table $T_0$.\\ 
\\            
\textbf{Case 5.1.3}
None of the block having elements from $S_2$ or $S_3$ coincides with the blocks having element from $S_1$ in the table $T_1$. In this case, we can have two cases. Either the block having the element from $S_3$ coincides with the block having the element from $S_2$, or it does not. \newline Let us first consider the case where a block having the element from $S_3$ coincides with one of the blocks having an element from $S_2$. Without loss of generality, we can say that block having the element $n_5$  coincides with the block having the element $n_4$. In this case, we send the block having the element $n_5$ to table $T_1$, and the rest of the empty block lying on the line, which contains this block to table $T_0$. We send the block having the element $n_4$ to table $T_0$, and the rest of the block lying on the line, which contains this block to table $T_1$. Now we see the positions of the blocks having an element from  $S_1$. Let us first consider the case where a block having the element from $S_1$ lies on the line, which contains blocks having elements from $S_2$. Without loss of generality, let us say block having the element $n_1$ lies on the line, which contains blocks having elements from $S_2$. In this case, we send the block having the element $n_1$ to the table $T_0$, and all other block lying on the line containing this block to table $T_1$. Rest all the empty blocks are sent to table $T_0$. On the other hand, if the blocks having elements from $S_1$ do not lie on the line which contains blocks having elements from $S_2$, then we send the block having elements from $S_1$ to table $T_1$, and rest all the empty blocks to table $T_0$. \newline If the block having an element from $S_3$ do not coincide with a block having an element from $S_2$ then we send all the block having elements to table $T_1$, and all the empty blocks to table $T_0$.\\
\\
\textbf{Case 5.2}
Now we consider the case where one of the sets having elements lie on a line and other set having elements lie on the different lines. Without loss of generality, let us consider the case where elements belonging to  $S_1$ lies on a line and the elements belonging to $S_2$ lies on the different lines.\\
\\            
\textbf{Case 5.2.1}
All the blocks having elements from $S_2$ and $S_3$  coincides with the block having an element from $S_1$. In this case, we send the block having an element from $S_2$ and $S_3$ to table $T_0$, and rest of the empty block lying on the lines containing these blocks to table $T_1$. Further, we send the blocks having elements from $S_1$ to table $T_1$, and the rest of the empty blocks lying on the line, which contains these blocks to table $T_0$. We send the rest of the empty blocks to table $T_0$.\\ 
\\        
\textbf{Case 5.2.2}
Two blocks having elements from $S_2$ and $S_3$ coincides with the block having an element from $S_1$. Now here we can have two cases either those two blocks have elements belonging to $S_2$ or we can have one element belonging to $S_2$ and other to $S_3$.\newline Let us first consider the case where two blocks having elements from $S_2$ coincides with a block having an element from $S_1$. without loss of generality, let us say that block having element $n_3$ coincides with the block having the element $n_1$, and the block having the element $n_4$ coincides with the block having $n_2$. In this case, we send the block having elements $n_3$ and $n_4$ to table $T_0$, and the rest of the empty block lying on the lines containing these blocks to table $T_1$. Now we see the position of the block having the element $n_5$. At this point, we can have two cases either the line which contains the block having the element $n_5$ passes through one of the block having an element from $S_1$ or it does not. Let us first consider the case where it passes through one of the blocks having elements from $S_1$. Without loss of generality, let us say that the line which contains block having the element $n_5$ passes through the block having the element $n_1$. In this case, we send the block having the element $n_5$ to table $T_0$, and the rest of the block lying on the line, which contains this block to table $T_1$. Further, we send the block having the element $n_1$ to table $T_0$, and the rest of the block lying on the line containing this block to table $T_1$. Rest all the empty blocks are sent to table $T_0$. On another hand, if the line which contains the block having the element $n_5$ do not pass through blocks having elements from $S_1$, then we send the block having the element $n_5$ to table $T_0$, and rest of the empty blocks on the line containing this block to table $T_1$. Furthermore, we send the blocks having elements from $S_1$ to table $T_1$, and the rest of the empty blocks to table $T_0$.\newline Now we consider the case where one of the blocks having an element from $S_2$, and another block having an element from $S_3$ coincides with the block having an element from $S_1$. Without loss of generality, say blocks having the element $n_3$ and $n_5$ coincides with a block(blocks) having an element(elements) from $S_1$. Now here we can have two cases, either $n_3$ and $n_5$ coincides with the same block having an element from $S_1$, or it coincides with different blocks having elements from $S_1$. Let us first consider the case where blocks having element $n_3$ and $n_5$ coincides with same block having an element from $S_1$, say $n_1$. In this case, we send the blocks having elements from $S_2$ and $S_3$ to table $T_0$, and the rest of the empty blocks lying on the lines containing these blocks to table $T_1$. Further, we send the block having the element $n_2$ to table $T_0$, and the rest of the block, which lies on the line containing this block to table $T_1$. Rest all the empty blocks are sent to table $T_0$. Now we consider the case where blocks having elements $n_3$ and $n_5$ coincides with different blocks having elements from $S_1$. Without loss of generality, let us say that block having element $n_3$ coincide with the block having the element $n_1$, and the block having the element $n_5$ coincide with the block having the element $n_2$. In this case, we send the blocks having elements from $S_2$ and $S_3$ to table $T_0$, and all the empty blocks lying on the line containing these blocks to table $T_1$. Further, we send the block having the element $n_2$ to table $T_0$, and rest of the block lying on the line which contains this block to table $T_1$.\\   
\\
\textbf{Case 5.2.3}
Only one block having element from $S_2$ or $S_3$ coincides with the block having an element from $S_1$ in table $T_1$. Now here we can have two cases, either the block having an element from $S_2$ coincides, or the block having an element from $S_3$ coincides with the block having an element from $S_1$.\\
\\ 
\textbf{Case 5.2.3.1} A block having an element from $S_2$ coincides with a block having an element from $S_1$ in table $T_1$. Without loss of generality, let us say that the block having the element $n_3$ coincides with the block having the element $n_1$. Now we see the position of the blocks having the element $n_4$ and $n_5$. Let us first consider the case where blocks having elements $n_4$ and $n_5$ coincide.\newline If the blocks having elements $n_4$ and $n_5$ coincides on the line which contains blocks having elements from $S_1$, then we send the blocks having elements $n_3,n_4$ and $n_5$ to table $T_0$, and all the blocks lying on the lines containing these blocks to table $T_1$. We send the block having element $n_1$ and $n_2$ to table $T_1$. Rest all the empty blocks are sent to table $T_0$.\newline If the blocks having elements $n_4$ and $n_5$ coincides outside the line which contains blocks having elements from $S_1$, then we send the blocks having elements $n_2,n_3$ and $n_4$ to table $T_0$, and all the blocks lying on the lines containing these blocks to table $T_1$. We send the blocks having elements $n_1$ and $n_5$ to table $T_1$. Rest all the empty blocks are sent to table $T_0$.\newline Now we are left with the case where blocks having elements $n_4$ and $n_5$ do not coincide. This can have several cases. Let us first consider the case where both the blocks having elements $n_4$ and $n_5$ lies on the line, which contains blocks having elements from $S_1$. In this case, we send the blocks having elements $n_3,n_4$, and $n_5$ to table $T_0$, and all the blocks lying on the lines containing these blocks to table $T_1$. We send the blocks having elements $n_1$ and $n_2$ to table $T_1$. Rest all the empty blocks are sent to table $T_1$.\newline Now let us consider the case where $n_4$ and $n_5$ do not lie on the line, which contains block having elements from $S_1$. In this case, we send the blocks having elements $n_3,n_4$, and $n_5$ to table $T_1$, and all the blocks lying on the lines containing these blocks to table $T_0$. Block having the element $n_1$ is sent to table $T_0$, and rest all the blocks lying on the line containing this block is sent to table $T_1$. Rest all the empty blocks are sent to table $T_0$.\newline Now, we can also have a case where only one block having the element $n_4$ or $n_5$ lies on the line, which contains blocks having elements from $S_1$. Let us first consider the case where block having the element $n_4$ lies on the line, which contains block having elements from $S_1$. In this case, we send the blocks having elements $n_3,n_4$, and $n_5$ to table $T_0$, and all other blocks lying on the line containing these blocks to table $T_1$. Now, we see whether the block having elements from $S_1$ lies on the line, which contains block having the element $n_5$. Without loss of generality, let us say block having the element $n_1$ lies on the line, which contains block having the element $n_5$. In this case, we send the block having the element $n_1$ to table $T_0$, and all other blocks lying on the line, which contains this block to table $T_0$. Rest all the empty blocks are sent to table $T_1$. On another hand, if the line which contains block having the element $n_5$ do not pass through block having element from $S_1$, then we send the blocks having elements from $S_1$ to table $T_1$. Rest all the empty blocks are sent to table $T_0$. Now let us consider the case where block having the element $n_5$ lies on the line, which contains block having elements from $S_1$. In this case, we send the blocks having elements $n_2,n_3,n_4$, and $n_5$ to table $T_0$, and all other blocks lying on the line containing these blocks to table $T_1$. Rest all the empty blocks are sent to table $T_0$.\\
\\            
\textbf{Case 5.2.3.2}
The block having element from $S_3$ coincides with block having element from $S_1$ in table $T_1$. Without loss of generality, let us say that block having element $n_5$ coincides with the block having the element $n_1$. Now we see the position of the blocks having elements from $S_2$.\newline Both the blocks having an element from $S_2$ lies on the line, which contains block having an element from $S_1$ in table $T_1$. In this case, we send the blocks having the elements $n_3,n_4$, and $n_5$ to table $T_0$, and all the blocks lying on the lines containing these blocks to table $T_1$. Blocks having elements from $S_1$ are sent to table $T_1$ and rest all the empty blocks are sent to table $T_0$.\newline Both the blocks having an element from $S_2$ do not lie on the line, which contains blocks having an element from $S_1$. In this case, we send the blocks having elements $n_3,n_4$, and $n_5$ to table $T_1$, and all the empty blocks lying on the lines containing these blocks to table $T_0$. Further, we send the block having the element $n_1$ to table $T_0$, and all the blocks lying on the line containing this block to table $T_1$. Rest all the empty blocks are sent to table $T_0$.\newline Now let us consider the case where only one block having the element from $S_2$ lies on the line, which contains blocks having elements from $S_1$ in table $T_1$. Without loss of generality, let us say that block having element $n_3$ lies on the line, which contains block having an element from $S_1$. In this case, we send the blocks having elements $n_4$ and $n_5$ to table $T_1$, and all the empty blocks lying on the lines which contain these blocks to table $T_0$. We send the blocks having the element $n_1$ and $n_3$ to table $T_0$, and rest all the blocks lying on the lines, which contain these blocks to table $T_1$. Rest all the empty blocks are sent to table $T_0$.\\
\\
\textbf{Case 5.2.4}
None of the blocks having element from $S_2$ or $S_3$ coincides with the block having elements from $S_1$ in table $T_1$. In this case, we can have either the block having the element $n_5$ coincides with the block having elements from $S_2$ or it does not. \newline Let us first consider the case where the block having the element $n_5$ coincides with one of the blocks having the element from $S_2$. Without loss of generality, let us say that the block having the element $n_5$ coincides with the block having the element $n_3$. Now we see whether the block having the element $n_4$ lies on the line, which contains block having elements from $S_1$ or not. Let us first consider the case where block having the element $n_4$ lies on the line, which contains block having an element from $S_1$. In this case, we send the block having the element $n_3$ to table $T_0$, and all the blocks lying on the line containing this block to table $T_1$. We send the block having the element $n_5$ to table $T_1$, and all the blocks lying on the line, which contains this block to table $T_0$. Now we see whether any of the block having elements from $S_1$ lies on the line having the element $n_3$ or not. Without loss of generality, let us say the block having the element $n_1$ lies on the line, which contains block having the element $n_3$. In this case, we send the blocks having the elements $n_1$ and $n_4$ to table $T_0$, and all other blocks lying on the lines which contain these blocks to table $T_1$. Rest all the empty blocks are sent to table $T_0$. If none of the blocks having elements from $S_1$ lies on the line, which contains block having the element $n_3$, then we send the blocks having elements $n_1,n_2$, and $n_4$ to table $T_1$, and rest all the empty blocks to table $T_0$. Now let us consider the case where block having the element $n_4$ do not lie on the line, which contains block having elements from $S_1$. In this case, the assignment made in the previous paragraph will work if we send the block having the element $n_4$ to table $T_1$, and all the blocks lying on the line containing this block to table $T_0$.
\newline Now, we are left with the case where block having the element $n_5$ do not coincide with the block having an element from $S_2$. In this case, we send the blocks having elements to table $T_1$, and all the empty blocks to table $T_0$.\\
\\        
\textbf{Case 5.3}
All the elements belonging $S_1,S_2$ and $S_3$ lies on the different line. In this case, we send the blocks having elements to table $T_0$ and all the empty blocks to table $T_1$.\\
\\
\textbf{Case 6}
Two elements belonging to $S_1=\{n_1,n_2\}$ lies in a same superblock, and the element $n_3,n_4$ and $n_5$ to the different superblocks. If the blocks having elements from $S_1$ lies on the different lines, then the assignment made in Case 5.3 can be used. So let us consider the case where blocks having the element from $S_1$ lies on the same line.\\
\\
\textbf{Case 6.1}
Only one block having element from $n_3,n_4$ and $n_5$ coincide with the block having element from $S_1$. Without loss of generality let us say block having the element $n_3$ coincide with the block having the element $n_1$.\\
\\        
\textbf{Case 6.1.1}
Blocks having elements  $n_4$ and $n_5$ coincide.\\ 
\\            
\textbf{Case 6.1.1.1}
Blocks having elements $n_4$ and $n_5$ coincides on the line which contains blocks having elements from $S_1$. In this case, we send the blocks having elements $n_3,n_4$, and $n_5$ to table $T_0$, and the rest of the empty blocks lying on the lines containing these blocks to table $T_1$. Blocks having the elements $n_1$ and $n_2$ are sent to table $T_1$, and rest all the empty blocks are sent to table $T_0$.\\  
\\
\textbf{Case 6.1.1.2}
Blocks having elements $n_4$ and $n_5$  coincides outside the line which contains element from $S_1$. Now here we can have several cases depending upon whether the lines which contain blocks having elements $n_4$ and $n_5$ passes through blocks having elements from $S_1$ or not. \newline Without loss of generality, let us say that line which contains block having the element $n_4$ passes through the block having the element $n_1$, and the line which contain block having the element $n_5$ passes through the block having the element $n_2$. Now in this case we send the blocks having elements $n_4$ and $n_1$ to table $T_1$. Further, we send the blocks having elements $n_3,n_5$ and $n_2$ to table $T_0$, and rest of the blocks lying on the lines containing these blocks to table $T_1$. We send rest of the empty blocks to table $T_0$.\newline Now consider a case where only one line which contains a block from $n_4$ or $n_5$ passes through the block having an element from $S_1$. Without loss of generality, let us say block having the element $n_4$ passes through the block having an element from $S_1$. Lets first consider the case where the line which contains block having the element  $n_4$ passes through the block having the element $n_2$. In this case, we send the blocks having elements $n_2,n_3,n_4$, and $n_5$ to table $T_0$, and rest of the blocks lying on the lines containing these blocks to table $T_1$. Further, we send the block having the element $n_1$ to table $T_1$, and the rest of the empty blocks to table $T_0$. Now without loss of generality, we can also have a case where block having the element $n_4$ passes through the block having the element $n_1$. In this case, we send the block having the elements $n_1,n_3,n_4$, and $n_5$ to table $T_0$, and the rest of the blocks lying on the lines containing these blocks to table $T_1$. We send the block having the element $n_2$ to table $T_1$, and the rest of the empty blocks to table $T_0$. If none of the lines which contains blocks having elements $n_4$ and $n_5$ passes through blocks having elements from $S_1$, then we send the blocks having elements $n_3,n_4$ and $n_5$ to table $T_0$, and all the blocks lying on the lines containing these blocks to table $T_1$. Further, we send the blocks having elements from $S_1$ to table $T_1$. Rest all the empty blocks are sent to table $T_0$. \\      
\\
\textbf{Case 6.1.2}
Blocks having the elements $n_4$ and $n_5$ do not coincide.\\
\\ 
\textbf{Case 6.1.2.1}
Both the blocks having element $n_4$ and $n_5$ lies on the line which contains element from $S_1$. In this case, we send the blocks having elements $n_3,n_4$, and $n_5$ to table $T_0$, and rest of the empty blocks lying on the line containing these blocks to table $T_1$. Further, we send the blocks having elements $n_1$ and $n_2$ to table $T_1$, and rest of the empty blocks to table $T_0$.\\ 
\\
\textbf{Case 6.1.2.2}
One of the blocks having an element  $n_4$ or $n_5$ lies on the line, which contains blocks having elements from  $S_1$ and other lies outside of it. Without loss of generality, let us say block having the element $n_4$ lies on the line, which contains blocks having elements from  $S_1$, and block having the element $n_5$ lies outside it. In this case, we send the blocks having elements $n_3,n_4$, and $n_5$ to table $T_0$, and all the blocks lying on the lines containing these blocks to table $T_1$. If any block having elements from $S_1$ lies on the line, which contains block having elements $n_5$, then we send that block to table $T_0$, and all other blocks lying on the line containing that block to table $T_1$. Rest all the empty blocks are sent to table $T_1$. On another hand, if none of the blocks having elements from $S_1$ lies on the line which contains block having elements $n_5$, then we send the block having elements from $S_1$ to table $T_1$, and rest all the empty blocks to table $T_0$.  \\
\\ 
\textbf{Case 6.1.2.3}
None of the blocks having elements $n_4$ or $n_5$ lies on the line, which contains blocks having elements from $S_1$. In this case we send the blocks having elements $n_2,n_3,n_4$ and $n_5$ to table $T_1$. Further, we send the block having the element $n_1$ to table $T_0$, and the rest of the empty blocks lying on the line, which contains this block to table $T_1$. Rest all the empty blocks are sent to table $T_0$.\\
\\
\textbf{Case 6.2}
Two blocks having elements from $n_3,n_4$ or $n_5$ coincides with the block having elements from $S_1$. Without loss of generality, let us say that blocks having elements $n_3$ and $n_4$ coincides with the block having elements from $S_1$. Now, here we can have two cases, either the blocks having elements $n_3$ and $n_4$ coincides with the same block having an element from $S_1$ or it coincides with the different blocks having elements from $S_1$. \newline Let us first consider the case where blocks having the element $n_3$ and $n_4$ coincides with different blocks having elements from $S_1$. Without loss of generality, let us say that block having element $n_3$ coincides with the block having the element $n_1$ and the block having the element $n_4$ coincides with the block having the element $n_2$. Now we see position of the block having element $n_5$. The block containing element $n_5$ can either lie on the line, which contains blocks having element from $S_1$ or not. Let us first consider the case where block having element $n_5$ do not lie on the line , which contains block having element from $S_1$. Further, let us consider the intersection of the line, which contains block having the element $n_5$ from the blocks having elements from $S_1$. Without loss of generality let us say that the line which contains block having element $n_5$ passes through the block having element $n_1$, in this case, we send the blocks having elements $n_1,n_3,n_4$ and $n_5$ to table $T_0$, and all other blocks lying on the lines which contains these blocks to table $T_1$. Rest all the empty blocks are sent to table $T_0$. On the another hand, if the block having element $n_5$ do not pass through the blocks having elements from $S_1$, then we send the blocks having $n_3,n_4$ and $n_5$ to table $T_0$, and all other blocks lying on the lines which contains these blocks to table $T_1$. Further, we send the blocks having elements from $S_1$ to table $T_1$, and rest all the empty blocks to table $T_0$. If the block having element $n_5$ lies on the line which contains blocks having elements from $S_1$, then we send the blocks having elements $n_3, n_4$ and $n_5$ to the table $T_0$, and all the empty blocks lying on these lines to table $T_1$. Further, we send the blocks having elements from $S_1$ to table $T_1$, and rest all the empty blocks to table $T_0$.\newline
Now we are left with the case where blocks having the element $n_3$ and $n_4$ coincides with only one block having an element from $S_1$. Without loss of generality, let us say that blocks having elements $n_3$ and $n_4$ coincides with the block having the element $n_1$. In this case, we see the position of the block having the element $n_5$. If the block having the element $n_5$ lies on the line which contains blocks having elements from $S_1$, then we send the block having the element $n_3,n_4$ and $n_5$ to table $T_0$, and all the empty blocks lying on the lines containing these blocks to table $T_1$. We send the blocks having elements from $S_1$ to table $T_1$. Rest all the empty blocks are sent to table $T_0$. Further, if the block having the element $n_5$ do not lie on the line which contains blocks having element from $S_1$ then we send the block having element $n_3,n_4$ and $n_5$ to table $T_0$, and all the empty blocks lying on the lines containing these blocks to table $T_1$. Now we see whether the line which contains block having the element $n_5$ passes through a block having an element from $S_1$ or not. Let us first consider the case where the line which contains block having the element $n_5$ passes through a block having an element from $S_1$. Without loss of generality, let us say that the line which contains block having the element $n_5$ passes through the block having the element $n_1$. In this case, we send the block having the element $n_1$ to table $T_0$, and all other blocks lying on the line containing this block to table $T_1$. Rest all the empty blocks are sent to table $T_0$. On another hand, if the line which contains block having the element $n_5$ do not pass through the block having an element from $S_1$, then we send the blocks having the elements from $S_1$ to table $T_1$. Rest all the empty blocks are sent to table $T_0$.\\        
\\
\textbf{Case 6.3} All the three blocks having elements $n_3,n_4$ and $n_5$ coincides with the blocks having elements from $S_1$. In this case, we send the blocks having elements $n_3,n_4$, and $n_5$ to table $T_0$, and all the blocks lying on the lines containing these blocks to table $T_1$. Blocks having elements from $S_1$ are sent to table $T_1$. Rest all the empty blocks are sent to table $T_1$.\\
\\ 
\textbf{Case 6.4}
None of the blocks having elements from $n_3,n_4$ and $n_5$ coincides with the blocks having elements from $S_1$.\newline Let us first consider the case where all the elements $n_3,n_4$ and $n_5$ coincides. If all of them coincides on the line which contains blocks having elements from $S_1$, then we send the blocks having elements $n_3,n_4$ and $n_5$ to table $T_0$, and rest of the empty blocks lying on the lines containing these blocks to table $T_1$. Further, we send the blocks having elements $n_1$ and $n_2$ to table $T_1$. Rest all the empty blocks are sent to table $T_0$.\newline If the blocks having elements $n_3,n_4$, and $n_5$ coincide outside the line, which contains blocks having elements from $S_1$, then we see the intersection of the lines containing these blocks with the blocks having elements from $S_1$. Now since all the blocks having elements from $n_3,n_4$, and $n_5$ coincides at most two lines having blocks containing these elements can intersect with the blocks having elements from $S_1$. Without loss of generality, let us say that the line which contains block having the element $n_3$ passes through the block having the element $n_1$, and the line which contains block having the element $n_4$ passes through the block having the element $n_2$. In this case we send the blocks having element $n_1$ and $n_3$ to table $T_1$. Further, we send the blocks having elements $n_2,n_4$, and $n_5$ to table $T_0$, and the rest of the empty blocks lying on the line containing these blocks to table $T_1$. Rest all the empty blocks are sent to table $T_0$.\newline Without loss of generality, let us now consider a case where only one line having element say $n_3$ passes through the block having element, say $n_1$, then we send the blocks having elements $n_1,n_2$ and $n_3$ to table $T_1$. Further, we send the blocks having elements $n_4$ and $n_5$ to table $T_0$, and the rest of the empty blocks lying on the lines containing these blocks to table $T_1$. Rest all the empty blocks are sent to table $T_0$.\newline If none of the lines containing blocks having elements $n_3,n_4$ and $n_5$ passes through the blocks having element from $S_1$, then we send the blocks having elements $n_3,n_4$ and $n_5$ to table $T_0$, and rest of the empty blocks lying on the line containing these blocks to table $T_1$. Further, we send the blocks having element $n_1$ and $n_2$ to table $T_1$, and rest of the empty blocks to table $T_0$.\newline Now let us consider the case where only two blocks having elements from $n_3,n_4$, and $n_5$ coincides. Without loss of generality, let us say the blocks having elements $n_3$ and $n_4$ coincide. Now we see the intersection of lines containing the blocks $n_3$ and $n_4$. Let us first consider the case where both the lines containing blocks $n_3$ and $n_4$ passes through the blocks having elements from $S_1$. Without loss of generality, let us say the line which contains block having the element $n_3$ passes through the block having the element $n_1$, and the line which contains block having the element $n_4$ passes through the block having the element $n_2$. In this case, we send the blocks having elements $n_1$ and $n_3$ to table $T_0$, and the rest of the blocks lying on the line containing these blocks to table $T_1$. We send the block having the element $n_2$ and $n_4$ to table $T_1$, and all other blocks lying on the lines containing these blocks to table $T_0$. Now, we see the position of the block having the element $n_5$. Kindly note that it is important to remember that no three lines passing through a point lie in the same plane in this case. If the block having the element $n_5$ lies on the line, which contains block having the element, then we send the block having the element $n_5$ to table $T_0$, and all other blocks lying on the line containing this block to table $T_1$. Rest all the empty blocks are sent to table $T_0$. \newline Now
let us consider the case where only one line which has block having the element $n_3$ or $n_4$ passes through the block having elements from $S_1$. Without loss of generality, let us say that block having the element $n_3$ passes through the block having the element $n_1$. In this case, we send the block having the element $n_3$ to table $T_1$, and all the blocks lying on the line containing these blocks to table $T_0$. The block having the element $n_4$ is sent to table $T_0$, and all the blocks lying on the line containing this block to table $T_1$. Now we see the position of the block having the element $n_5$. Let us first consider the case where block having the element $n_5$ lies on the line, which contains blocks having elements from $S_1$. In this case, if the line which contains block having the element $n_5$ passes through the block having the element $n_3$ then we send the block having the element $n_5$ to table $T_1$, and all the block lying on the line containing this block to table $T_0$. The blocks having elements from $S_1$ are sent to table $T_1$, and all the blocks lying on the line containing these blocks to table $T_0$. Rest all the empty blocks are sent to table $T_0$. On another hand, if the line which contains block having the element $n_5$ do not pass through the block having the element $n_3$ then we send the block having $n_5$ to table $T_0$, and all the blocks lying on the line containing this block to table $T_1$. Further, the blocks having elements from $S_1$ are sent to table $T_1$, and all the blocks lying on the line containing these blocks to table $T_0$. Rest all the empty blocks are sent to table $T_0$. If the block having the element $n_5$ do not lie on the line which contains blocks having elements from $S_1$, then we send the block having the element $n_5$ to table $T_0$, and all the blocks lying on the line containing this block to table $T_1$. Now we see whether the block having elements from $S_1$ lies on the line, which contains block having the element $n_5$ or not. Without loss of generality, let us say that block having element $n_1$ lies on the line, which contains block having the element $n_5$. In this case, we send the block having the element $n_1$ to table $T_0$, and all the blocks lying on the line containing this block to table $T_1$. Rest all the empty blocks are sent to table $T_0$. On the other hand, if none of the blocks having element from $S_1$ lies on the line which contains block having element $n_5$, then we send the block having elements from $S_1$ to table $T_1$, and all the empty blocks lying on the line containing these blocks to table $T_0$. Rest all the empty blocks are sent to table $T_0$.\newline If none of the lines which contains block having elements $n_3$ or $n_4$ passes through block having element from $S_1$, then we send the blocks having elements $n_3,n_4$ and $n_5$ to  table $T_0$, and all the empty blocks lying on the lines containing these blocks to table $T_1$. Now, either the block having an element from $S_1$ lies on the line which contains block having the element $n_5$ or it does not. Let us first consider the case where block having elements from $S_1$ lies on the line, which contains block having the element $n_5$. Without loss of generality, let us say that block having element $n_1$ lies on the line, which contains block having the element $n_5$. In this case, we send the block having the element $n_1$ to table $T_0$, and all the blocks lying on the line containing this block to table $T_1$. Rest all the empty blocks are sent to table $T_0$. If none of the block having an element from $S_1$ lies on the line, which contains block having the element $n_5$, then we send the blocks having an element from $S_1$ to table $T_1$. Rest all the empty blocks are sent to table $T_0$. \newline Now we are left with the case where none of the blocks having elements from $n_3,n_4$, and $n_5$ coincides. In this case, we send the blocks having elements to table $T_1$, and the rest of the empty blocks to table $T_0$.\\         
\\            
\textbf{Case 7}
All the blocks having elements  $n_1,n_2,n_3,n_4$ and $n_5$ lies in the different  superblocks. In this case, we send all the blocks having elements to table $T_0$, and all the empty blocks to table $T_1$.
\end{document}